\documentclass[10pt, conference, letterpaper]{IEEEtran}
\IEEEoverridecommandlockouts
\usepackage{cite}
\usepackage{amsmath,amssymb,amsfonts,amsthm}
\usepackage{algorithmic}
\usepackage{graphicx}
\usepackage{textcomp}
\usepackage{xcolor}
\def\BibTeX{{\rm B\kern-.05em{\sc i\kern-.025em b}\kern-.08em
    T\kern-.1667em\lower.7ex\hbox{E}\kern-.125emX}}

\usepackage{subfig}
\usepackage{multirow}
\usepackage{capt-of}

\usepackage{hyperref}

\newtheorem{newrule}{Rule}
\newtheorem{theorem}{Theorem}

\IEEEaftertitletext{\vspace{-2.5\baselineskip}}

\begin{document}

\title{\textsc{Prophet}: Conflict-Free Sharding Blockchain via Byzantine-Tolerant Deterministic Ordering\vspace{-13pt}}

\author{\IEEEauthorblockN{Zicong Hong$^{1}$, Song Guo$^{1, 2}$, Enyuan Zhou$^{1}$, Jianting Zhang$^{3}$,\\ Wuhui Chen$^{4}$, Jinwen Liang$^{1}$, Jie Zhang$^{1}$, and Albert Zomaya$^{5}$}
    \IEEEauthorblockA{$^1$Department of Computing, The Hong Kong Polytechnic University, Hong Kong $^2$The Hong Kong Polytechnic University\\ 
    Shenzhen Research Institute, China $^3$Computer Science Department, Purdue University, USA\\ $^4$School of Computer Science and Engineering, Sun Yat-sen University, and Pengcheng Laboratory, Shenzhen, China\\
    $^5$School of Computer Science, The University of Sydney, Australia\\
	zicong.hong@connect.polyu.hk, song.guo@polyu.edu.hk, en-yuan.zhou@connect.polyu.hk, zhan4674@purdue.edu\\
	chenwuh@mail.sysu.edu.cn, \{jinwen.liang, jie1zhang\}@polyu.edu.hk, albert.zomaya@sydney.edu.au}

\thanks{
This research was supported by fundings from the Key-Area Research and Development Program of Guangdong Province under grant No. 2021B0101400003,  Hong Kong RGC Research Impact Fund (RIF) with the Project No. R5060-19, General Research Fund (GRF) with the Project No. 152221/19E, 152203/20E, 152244/21E, and 152169/22E, the National Natural Science Foundation of China 61872310, Shenzhen Science and Technology Innovation Commission (JCYJ20200109142008673), and the National Key Research and Development Plan (2021YFB2700302). 
}
}

\maketitle



\begin{abstract}
Sharding scales throughput by splitting blockchain nodes into parallel groups. 
However, different shards' independent and random scheduling for cross-shard transactions results in numerous conflicts and aborts, since cross-shard transactions from different shards may access the same account. 
A deterministic ordering can eliminate conflicts by determining a global order for transactions before processing, as proved in the database field. 
Unfortunately, due to the intertwining of the Byzantine environment and information isolation among shards, there is no trusted party able to predetermine such an order for cross-shard transactions.
To tackle this challenge, this paper proposes \textsc{Prophet}, a conflict-free sharding blockchain based on \emph{Byzantine-tolerant deterministic ordering}. 
It first depends on untrusted self-organizing coalitions of nodes from different shards to pre-execute cross-shard transactions for prerequisite information about ordering. 
It then determines a trusted global order based on stateless ordering and post-verification for pre-executed results, through shard cooperation. 
Following the order, the shards thus orderly execute and commit transactions without conflicts. 
\textsc{Prophet} orchestrates the pre-execution, ordering, and execution processes in the sharding consensus for minimal overhead.
We rigorously prove the determinism and serializability of transactions under the Byzantine and sharded environment. 
An evaluation of our prototype shows that \textsc{Prophet} improves the throughput by $3.11\times$ and achieves nearly no aborts on 1 million Ethereum transactions compared with state-of-the-art sharding.
\end{abstract}

\vspace{-15pt}


\section{Introduction}

The prosperity of blockchain facilitates 
decentralized applications (dApps),
e.g., decentralized exchanges~\cite{werner2021sok} and non-fungible token~\cite{wang2021nonfungible}.
In Ethereum, 
the number of contract calls per day has more than tripled to over 3 million from Jun. to Sep. in 2021~\cite{triple}. However, the poor scalability of the existing blockchains, 
7 transaction per seconds (TPS) in Bitcoin~\cite{bitcoin} and 
15-45 TPS in Ethereum~\cite{ethereum}, cannot satisfy this growing demand for dApps. 
This is because their consensus requires all nodes to validate and execute every transaction, which aggravates the scalability problem and restricts smart contracts from more users and the dApps with more complex logic.

Sharding is one of the most promising technologies for scalability~\cite{sok_sharding}. 
It divides nodes into multiple consensus groups called \emph{shards} and distributes transactions to shards to process in parallel. The technology has been paid close attention by the academia~\cite{elastico, omniledger, rapidchain, monoxide, pyramid, pldi, byshard, brokerchain,sstore}. For the industry, some blockchains are being or have been upgraded to a sharding architecture, such as Zilliqa~\cite{zilliqa} and Eth2 upgrade in Ethereum~\cite{layer2}. 


\begin{figure}[t]
	\centering
	\subfloat[][]{
		\begin{minipage}[t]{0.31\linewidth}
			\centering
			\includegraphics[width=\linewidth]{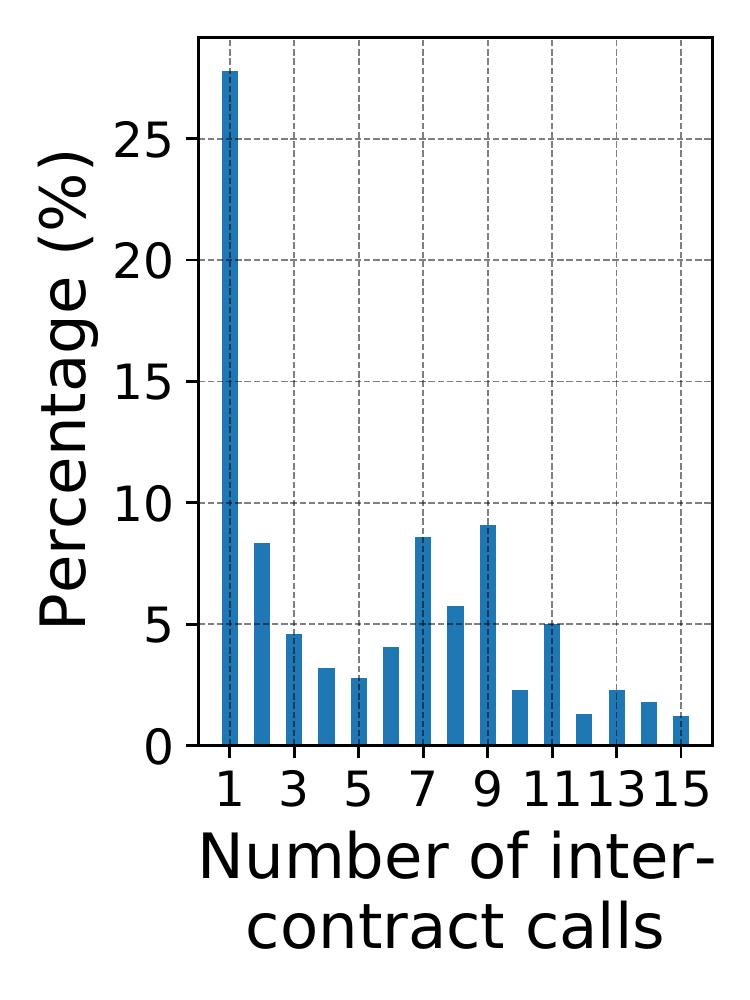}
		\end{minipage}
		\label{fig:inter_contract_call}
	}
	\subfloat[][]{
		\begin{minipage}[t]{0.65\linewidth}
			\centering
			\includegraphics[width=\linewidth]{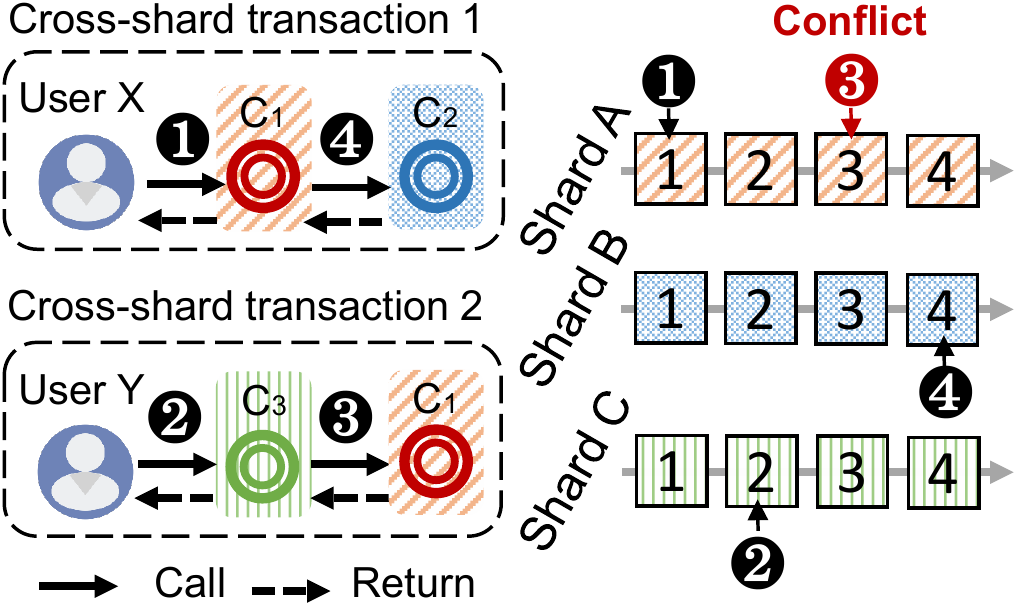}
		\end{minipage}%
		\label{fig:race_event}
	}
	\centering
	\caption{(a) Percentage of Ethereum transactions with different number of inter-contract calls from Oct. 2020 to May. 2021, (b) Illustration for conflicting cross-shard transactions. Each concentric circle represents a smart contract. Three contracts are located in three different shards. Each circled number represents the round at which a sub-transaction is committed.}
    \vspace{-20pt}
\end{figure}

While increasing the throughput in proportion to the number of shards, sharding technology introduces \emph{cross-shard transactions}, which are transactions involving the smart contracts in multiple shards.
More seriously, similar to the current software composed of numerous programs, a dApp often requires the cooperation of several contracts, significantly increasing the number and complexity of cross-shard transactions. 
As shown in \autoref{fig:inter_contract_call}, more than 70\% of Ethereum~\cite{ethereum} transactions for smart contracts include more than $2$ inter-contract calls and the average number is $8.94$. 
After introducing sharding, the contracts are located in different shards thus the inter-contract calls result in massive cross-shard transactions. 
To guarantee the atomicity and consistency of cross-shard transactions, each sharding system needs a cross-shard mechanism. 
The mechanism divides each cross-shard transaction into several sub-transactions, each of which corresponds to a shard, and commits each sub-transaction to the corresponding shard.

\emph{Unfortunately, we observe that the existing cross-shard mechanisms perform well below expectations in practice. 
As evaluated in \autoref{sec:evaluation}, an existing blockchain with 32 shards performs even worse than a non-sharding one when meeting the Ethereum transactions. }
It upends our usual understanding on sharding and drives us to reconsider the technology of sharding.
This poor performance results from inherent conflicts among cross-shard transactions and the independent and random scheduling for cross-shard transactions in different shards.
In particular, multiple cross-shard transactions may access the state of the same smart contract (e.g., Cxn.\footnote{In this paper, for simplification, we use Cxn. to denote cross-shard transaction, Txn. to denote single-shard transaction, and Blk. to denote block.} 1 and 2 access Contract $C_1$ in \autoref{fig:race_event}). 
Although Cxn.~1 is issued first, the state of $C_1$ can be modified by Cxn.~2 before Cxn.~1 is committed. 
The existing sharding systems mainly adopt two-phase locking (2PL) (e.g.,~\cite{omniledger,sigmod_sharding,chainspace}) or optimistic concurrency control (OCC) (e.g.,~\cite{monoxide,pyramid}) to avoid conflict and guarantee serializability of transactions.
However, as evaluated in \autoref{sec:evaluation}, more than 50\% of transactions are aborted or rolled back due to race conditions, because real smart contract workloads contain frequent read-write conflicts. Therefore, both 2PL and OCC exhibit high abort rates and strongly limit the performance of blockchain sharding.

To eliminate the high abort rates caused by non-deterministic race conditions, an intuitive idea is to introduce a predetermined serial global order for pending transactions before processing them.
The idea is inspired by distributed and deterministic database with a sequencing layer, which collects all database transactions for producing a global order before database execution~\cite{calvin,aria,prescient, pwv, deterministic_overview}. 
However, the sequencing layer in distributed databases is designed with a strict assumption that the layer contains trusted machines for storing the whole database state, which is far from trivial for the blockchain. 
Due to the intertwining of the information isolation among shards (i.e., each node only stores a proportion of contracts) and Byzantine environment (i.e., blockchain nodes do not trust each other), there are no trusted nodes to predetermine such an order for transactions involving the state of different shards.

Therefore, we propose \textsc{Prophet}, a conflict-free sharding blockchain, based on a new idea named \emph{Byzantine-tolerant deterministic ordering}. 
Specifically, to overcome the challenge of information isolation, the nodes from different shards are allowed to form self-organizing coalitions to pre-execute pending transactions, including single-shard and cross-shard ones, for prerequisite information about ordering. 
Then, a random shard is delegated to sequence the pre-executed transactions for a global order based on prerequisite information. 
To deal with Byzantine failures in blockchain, a shard-cooperation proof sharing is proposed to verify and correct untrusted pre-execution results without interruption of consensus.
With such an order, transactions in different shards can be executed and committed orderly without conflicts.
The new architecture will not break any decentralization principle of blockchain.

The contributions of this work are summarized as follows.

\begin{itemize}
    \item We propose an idea of Byzantine-tolerant deterministic ordering and develop a conflict-free sharding blockchain named \textsc{Prophet}, minimizing the number of transaction aborts caused by non-deterministic contract contention. 
    \item On top of the shards, \textsc{Prophet} introduces two new types of parties, i.e., sequence shard and reconnaissance coalition, with a little additional overhead and designs a new cooperative consensus for a global order based on the cooperation and joint supervision among shards.
    \item Based on the characteristics of smart contracts, such as inter-contract calls and contract instructions, we present the designs of fine-grained ordering, asynchronous correction, and parallel pre-execution for efficiency.
    \item We develop a prototype for \textsc{Prophet} and conduct a comprehensive evaluation. 
    \textsc{Prophet}
    improves the throughput by $3.11\times$ (i.e., 1203 TPS) on 1 millions Ethereum transactions compared with state-of-the-art sharding systems.
\end{itemize}

\section{Background}
\label{sec:background}

\begin{figure*}[t]
	\centering
	\begin{minipage}{.65\linewidth}
	    \centering
    	\subfloat[][OCC-based cross-shard mechanism]{
    		\begin{minipage}[t]{0.98\linewidth}
    			\centering
    			\includegraphics[width=\linewidth]{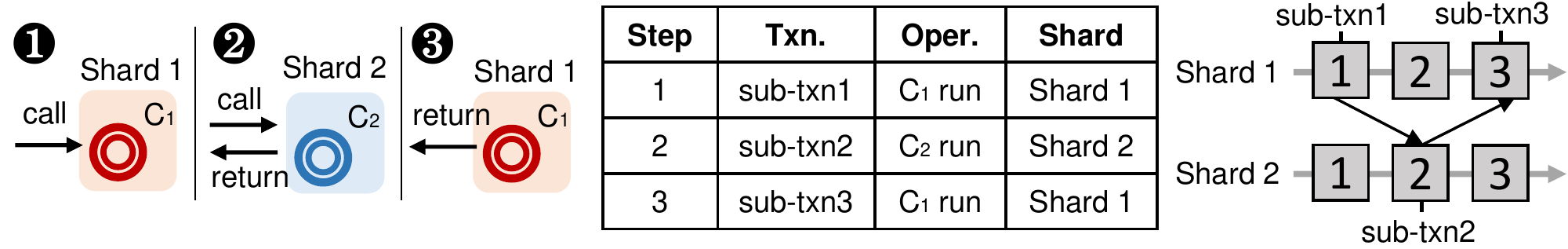}
    		\end{minipage}
    		\label{fig:related_work_1}
    	}
    \\\vspace{-8pt}
	\subfloat[][2PL-based cross-shard mechanism]{
		\begin{minipage}[t]{0.98\linewidth}
			\centering
			\includegraphics[width=\linewidth]{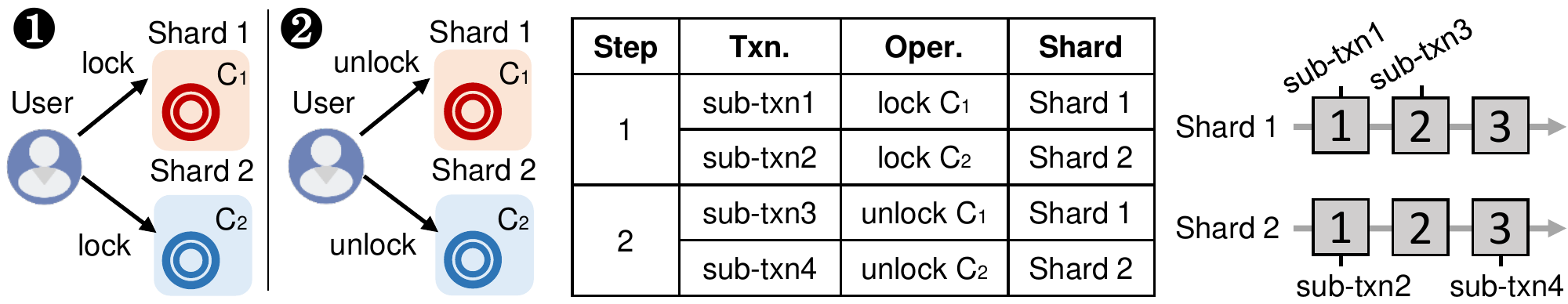}
		\end{minipage}%
		\label{fig:related_work_2}
	}
    \caption{Illustration for the existing cross-shard mechanism for smart contracts.}
	\end{minipage}
	\hfill
	\begin{minipage}{.34\linewidth}
    	\centering
    	\renewcommand{\arraystretch}{1.2}
        \begin{tabular}{l | c}
            \hline
            \textbf{System} & {\textbf{Type}} \\ \hline \hline
            Elastico~\cite{elastico}, CoSplit~\cite{pldi} & None \\\hline
            Omniledger~\cite{omniledger}, ByShard~\cite{byshard}  & \multirow{3}{*}{2PL} \\ 
            Chainspace~\cite{chainspace}, Rapidchain~\cite{rapidchain} & \\ 
            Dang \textit{et al.}~\cite{sigmod_sharding} & \\\hline
            Monoxide~\cite{monoxide}, Pyramid~\cite{pyramid} & OCC \\\hline
        \end{tabular}
        \vspace{+5pt}
        \captionof{table}{Cross-shard mechanisms of different sharding systems. (\cite{elastico,pldi} do not have cross-shard mechanisms since they do not achieve sharding for storage, i.e., there are no cross-shard transactions.)}
        \label{tab:comparison}
    \end{minipage}
	\label{fig:related_work}
    \vspace{-18pt}
\end{figure*}

Elastico~\cite{elastico} is the first public blockchain for transaction verification sharding, but it requires nodes to store the entire ledger. Omniledger~\cite{omniledger} and RapidChain~\cite{rapidchain} achieve state sharding, which means each shard only needs to store a proportion of the ledger. Different from the previous works for UTXO transactions, Monoxide~\cite{monoxide} is a blockchain sharding for account-based transactions. Some recent research efforts include smart contract sharding~\cite{chainspace, icde}, sharding blockchain database~\cite{blockchaindb}, dynamic sharding~\cite{skychain}, and permissioned blockchain sharding~\cite{sigmod_sharding, sharper}.

Although sharding improves the scalability of blockchain, it raises a new challenge to the concurrency control of cross-shard transactions. 
Each cross-shard transaction is a transaction involving multiple contracts distributed in different shards. 
To commit such a cross-shard transaction, every sharding system requires a cross-shard mechanism to guarantee its atomicity and consistency for the related shards. 
Since the cross-shard transactions from different shards may access the state of the same contracts, the cross-shard mechanism needs to resolve their conflicts. 
As shown in Table \ref{tab:comparison}, the existing cross-shard mechanisms can be classified into two classes, i.e., OCC and 2PL, which are discussed as follows.

For an OCC cross-shard mechanism, cross-shard transactions are committed in an optimistic manner. In detail, each cross-shard transaction is divided into multiple sub-transactions, each of which involves the contracts stored by the same shard.  
For example, as shown in \autoref{fig:related_work_1}, Cxn.~1 in \autoref{fig:race_event} is divided into three sub-transactions. 
We emphasize that in the scenario of smart contracts, the later sub-transaction can be known only after the former sub-transaction is executed. 
Thus, the sub-transactions need to be generated and committed in series. 
The basic premise of OCC is that most transactions do not conflict with other transactions. 
However, before all sub-transactions of a cross-shard transaction are committed, if its touched contracts are updated by other transactions, it needs to be aborted and the state change of the touched contracts will be withdrawn.

For a 2PL cross-shard mechanism, a cross-shard transaction is committed after requesting and releasing the related contracts' locks. 
A 2PL cross-shard mechanism assumes that clients or reference shards store the latest state of all contracts, called \emph{client-driven} and \emph{shard-driven}, respectively. 
Before committing a cross-shard transaction, the client (or the reference shard) pre-executes it and knows all its related contracts. 
Then, the client requests the locks from the shards storing the related contracts. 
Only after the client requests all locks, can the client be committed in the related shards and release all locks. 
The request and release for locks are realized by sub-transactions. 
For example, in \autoref{fig:related_work_2}, the client first requests the locks from Shard 1 and 2. 
It then commits the transaction and releases the locks. 
To avoid deadlocks, the 2PL cross-shard mechanisms need some deadlock prevention designs. 
For example, if a client (or reference shard) requests a lock of a contract that has been locked by other clients (or reference shards), it needs to release its acquired locks.

Although the cross-shard mechanisms guarantee that each cross-shard transaction can be committed in all its related shards, each shard independently packs and orders its transactions in each round. It results in non-deterministic race conditions and frequent aborting of cross-shard transactions since their involved contracts may be modified or locked by the others before they are committed, i.e., before all of its sub-transactions are committed. This makes the sharding systems perform even worse than the non-sharding ones.

\section{Strawman: An Ideal Cross-Shard Mechanism}
\label{sec:strawman}

We first describe an ideal cross-shard mechanism that ignores the decentralized and Byzantine environment of blockchain. It motivates the main design of \textsc{Prophet} in \autoref{sec:deterministic_sharding}. This mechanism guarantees that no transactions are aborted.

\begin{figure}[t]
    \centering
    \includegraphics[width=0.9\linewidth]{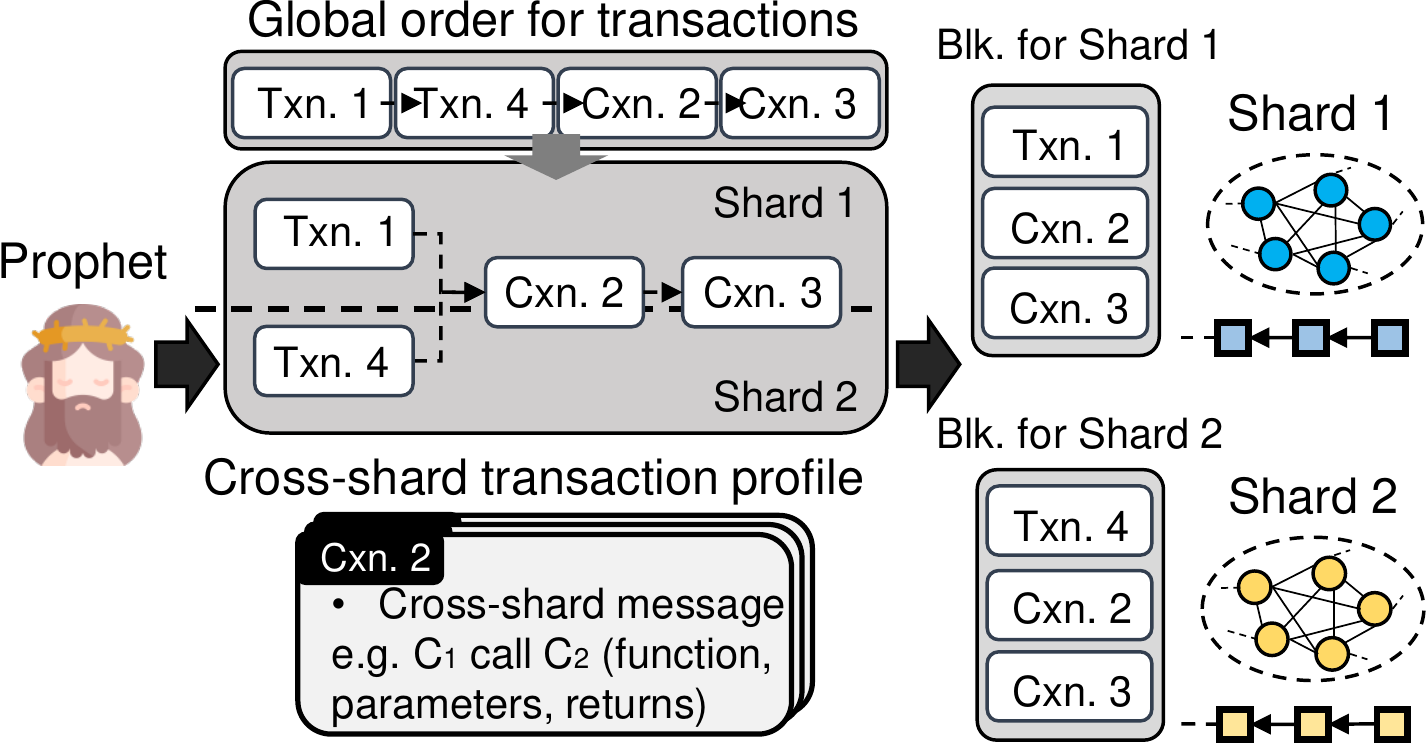}
    \caption{An ideal cross-shard mechanism for sharding.}
    \label{fig:prophet}
    \vspace{-18pt}
\end{figure}

As shown in \autoref{fig:prophet}, the mechanism introduces a new node called \emph{prophet}. Assume that the prophet is fully trusted and has infinite computing power and storage capacity. 
For each consensus round, it first validates and executes transactions sequentially for a global order of the transactions, such as $1423$ in the figure.
Next, according to the global order, the prophet generates a serial order for each shard's block. 
Specifically, for a single-shard transaction, since it only reads or writes the state of contracts in a shard, its execution only depends on the latest executed transaction in the shard. 
For a cross-shard transaction, it involves the state of multiple shards, thus its execution depends on several transactions from different shards. 
For example, in \autoref{fig:prophet}, Cxn.~2 executes following Txn.~1 and Txn.~4. 
After validation and execution, the blocks generated by the prophet correspond to a global order which shows the data dependency of all transactions in this round.

Moreover, the prophet can record some meta information required for the execution of each cross-shard transaction in each shard, called \emph{cross-shard transaction profile}.
The profile includes all cross-shard inter-contract calls and their parameters and returns, called \emph{cross-shard message}.
The profile enables each shard to execute transactions without communicating with the other shards during execution.

Finally, the prophet sends the serial order and the profile to the corresponding shard with a signature. 
The shards only need to replicate and execute transactions one after the other based on their received serial orders. 
All transactions can be committed without any conflicts according to the global order. 

Although the mechanism eliminates transaction aborts and requires minimal coordination among shards, the assumption of such a special node is too ideal. In particular, a fully trusted node violates the decentralized inherence of public blockchain and it has to locally maintain the whole state for all shards to pre-execute transactions. 
Thus, it raises a question about how to implement such a prophet in the decentralized and Byzantine environment of blockchain sharding. 
In the following, we present \textsc{Prophet}, which achieves a similar effect through the cooperation and supervision among shards in a distributed manner and without any trusted third party.

\section{Byzantine-Tolerant Deterministic Ordering for Blockchain Sharding}
\label{sec:deterministic_sharding}

\subsection{System Model \& Threat Model}
\label{sec:system_model}

\begin{figure*}[t]
    \centering
    \includegraphics[width=0.82\linewidth]{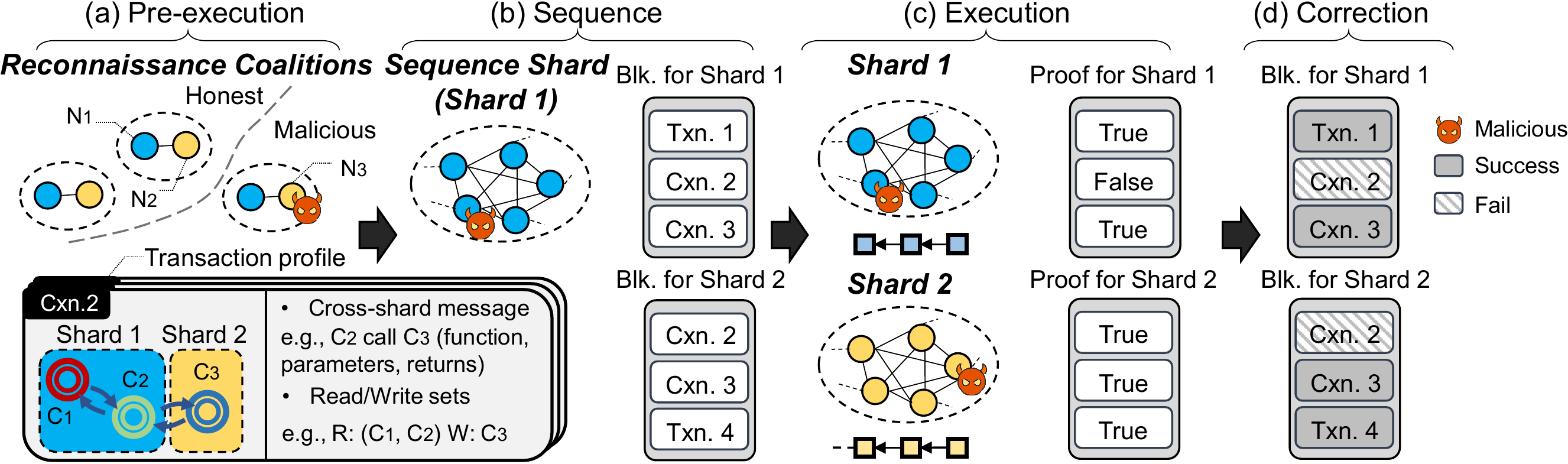}
    \caption{The architecture of \textsc{Prophet}.}
    \label{fig:prophet_overview}
    \vspace{-18pt}
\end{figure*}

Similar to the existing blockchain sharding~\cite{elastico,omniledger,rapidchain,pyramid}, \textsc{Prophet} proceeds in \emph{epochs}, each of which includes multiple rounds. 
It consists of a set of nodes following the Byzantine failure model which includes two kinds of nodes, i.e., \emph{honest} and \emph{malicious}. The honest nodes abide by all protocols. The malicious nodes are controlled by a Byzantine adversary and may collide with each other and violate the protocols in arbitrary manners, e.g., denial of service, tampering, and forgery. 
All nodes are randomly divided into multiple shards. The nodes in each shard store a proportion of contracts and verify and execute the transactions involving the stored contracts.
Besides, each lightweight client only stores its accounts and do not store any contracts. 
Since the world state of current blockchain is huge (e.g., total size in Ethereum is currently more than $130$ GB and keeps increasing~\cite{worldstate}), in \textsc{Prophet}, neither a node or a client can store the state of all contracts.

\subsection{Motivation \& Overview}
\label{Motivation}

We start with the problems that the strawman in \autoref{sec:strawman} highlights and build up our design step by step.


According to \autoref{sec:strawman}, the function of the prophet includes two tasks, i.e., pre-execution and ordering for pending transactions. 
Specifically, the pre-execution task requires the storing of the blockchain state, while the ordering task does not require it. 
It is because the pre-execution needs to output the prerequisite information about ordering (i.e., read/write sets and cross-shard messages) by executing transactions based on the blockchain state.
In comparison, the ordering is based on the prerequisite information provided by the pre-execution and does not need to execute any transaction or store any contract; thus, it is \emph{stateless}.


Based on their characteristics (i.e., requirements and workload), we delegate these two tasks to different parties as follows.
First, since any node cannot store the state of all contracts (as discussed in \autoref{sec:system_model}), to pre-execute all possible transactions, we introduce a new type of parties named \emph{reconnaissance coalitions} in which nodes from different shards cooperate with each other to pre-execute pending transactions (see \autoref{sec:pre_execution}).
Nodes can freely form or dissolve reconnaissance coalitions, which are off-chain.
\textsc{Prophet} distributes a proportion of transaction fees to reconnaissance coalitions with successfully committed pre-executed transactions.
(A more detailed analysis of incentives in reconnaissance coalitions will be left as our future work.)
Second, the task of ordering is stateless and does not need intensive computation; thus, any node or shard can do it. 
For security, we let every shard take on the ordering task in turn, and we call the shard responsible \emph{sequence shard} (see \autoref{sec:sequence}). The sequence shard will be updated in each epoch for load balancing.
With the help of these new parties, \textsc{Prophet} proceeds in four phases, i.e., \emph{pre-execution}, \emph{sequence}, \emph{execution}, and \emph{correction}, for each round.

\subsection{Phase 1: Pre-execution}
\label{sec:pre_execution}

During this phase, each reconnaissance coalition selects a disjoint set of pending transactions with some specific range of transaction hash and executes them one by one.
In each reconnaissance coalition, if a node meets an inter-contract call to a contract located in another shard when pre-executing a transaction, it turns to the other nodes in the same reconnaissance coalition. 
For example, in \autoref{fig:prophet_overview}, when pre-executing Cxn.~2, since Contract $C_3$ belongs to Shard 2, Node $N_1$ needs to send the function call and parameters to and get the returns from Node $N_2$. 
Thus, each reconnaissance coalition can be regarded as an individual able to execute the single-shard or cross-shard transactions for its related shards.

In particular, each transaction reads from the current state of the blockchain, executes its logic, and keeps the writes in a local write set. Since the change of each transaction is kept in local write set, the state read by each transaction is always the same. For each transaction, the reconnaissance coalition records its read/write set and its cross-shard messages during execution. To clarify our basic idea, we define the read/write set of a transaction as the addresses of smart contracts, which will be extended to a fine-grained one in \autoref{sec:fine_grained_ordering}. Moreover, the serial execution will be extended to a more efficient one in \autoref{sec:communication}. The cross-shard message about a cross-shard inter-contract call includes its function name, parameters and return.

We emphasize that there is no guarantee that all reconnaissance coalitions are honest because their formation is free and cannot guarantee that all nodes in a reconnaissance coalition are honest. 
In a Byzantine environment, the read/write sets and cross-shard messages recorded by a reconnaissance coalition with malicious nodes will be wrong. For example, in \autoref{fig:prophet_overview}, there is a malicious reconnaissance coalition since there is a malicious node $N_3$ in the coalition.
Since each successful pre-execution gains transaction fees, we assume that honest nodes tend to form and stay in reconnaissance coalitions involving the other honest nodes and leave coalitions involving malicious nodes (the detection of malicious nodes in a reconnaissance coalition will be discussed in \autoref{sec:correction}).

\subsection{Phase 2: Sequence}
\label{sec:sequence}

After a reconnaissance coalition pre-executes a transaction, it passes the transaction and the corresponding transaction profile to the sequence shard. 
The sequence phase starts when the leader of the sequence shard receives enough (i.e., more than a predefined threshold) transactions.

For each round, in the sequence phase, based on the pre-execution results passed by the reconnaissance coalitions, the leader of the sequence shard can determine which transactions will be included in the global order.
To avoid transaction conflicts, the sequence shard only allows the transactions involving disjoint read/write sets to be packed into the order. 
In such an order, the transactions are processed as the same as they are in the pre-execution phase since they are not in conflict. 
Based on the order, the sequence shard generates a serial order for the block of every shard and provides a transaction profile for transactions included in the block. 
For example, in \autoref{fig:prophet_overview}, the sequence shard proposes two new blocks to Shard 1 and 2, respectively.
Finally, based on the intra-shard consensus, the nodes in the sequence shard send the new blocks with a collective signature (such as CoSi~\cite{cosi} or BLS~\cite{bls} in the existing sharding~\cite{omniledger,pyramid}) to the shards.

The sequence shard has two characteristics. The first one is stateless, which means the nodes in the sequence shard can generate a global order depending on transaction profiles received from reconnaissance shards and without storing the state of all contracts and pre-executing transactions. 
The second one is trusted, which means each message published by the sequence shard is via an intra-shard consensus.

\subsection{Phase 3: Execution}

As discussed in \autoref{sec:pre_execution}, in the pre-execution phase, the pre-execution results cannot be guaranteed because the reconnaissance coalitions may be malicious. 
Moreover, in the sequence phase, because the sequence shard is stateless and only responsible for ordering the transactions instead of validating, there can be invalid transactions or conflict events existing in the order. Therefore, the shards need to execute and validate the transactions included in the received blocks and compare them with the read/write sets and transaction profile. 

During the execution, each shard runs an intra-shard consensus and executes all transactions based on the transaction profile. 
For example, in \autoref{fig:prophet_overview}, Shard 2 executes Cxn.~2 based on the function call from Contract $C_2$ to $C_3$ and the corresponding parameters in the transaction profile. 
If a shard finds that the read/write sets or cross-shard messages of a transaction are different from those provided by the reconnaissance coalitions, it can mark the transaction as invalid. 
However, if the execution results of a transaction exactly match its transaction profile, the shard can mark it as valid. 
For example, Shard 1 marks Cxn.~2 as invalid if the cross-shard message or read/write set in the transaction profile is incorrect.
The intra-shard consensus can guarantee that the result published by any shards is trusted.

\subsection{Phase 4: Correction}
\label{sec:correction}

The confirmation of a cross-shard transaction, i.e., a shard commits the transaction and updates the state of contracts based on the transaction, requires the proof generated by all the related shards of the transaction in the execution phase.
At the end of each round, 
every shard shares its validation results (i.e., proof generated in the execution phase) with the other shards. 
The proof denotes the validity of each pre-executed transaction included in the global order. 
Each cross-shard transaction can be committed in a shard only when the shard receives the validity proof from all the other shards related to the transaction. 
For example, in \autoref{fig:prophet_overview}, Cxn.~2 is related to Shard 1 and 2, thus it cannot be committed without the proof of both these two shards. 
Since the proof of Shard 1 marks it as invalid, it will not be committed. 
In addition, the honest nodes in the reconnaissance coalition responsible for Cxn.~2 can leave the coalitions and mark the nodes responsible for the invalid part (contracts in Shard 1) as malicious.

\subsection{Discussion}

Different from the traditional blockchain sharding that only has the execution phase, \textsc{Prophet} has three additional phases (The overhead will be analyzed in \autoref{sec:analysis}).
In each round, a deterministic global order for all transactions, including single-shard and cross-shard ones, can be generated and shared by all shards through these three phases. 
Following the order, the shards can orderly execute and commit transactions and update the blockchain state without conflicts.
The cooperation within reconnaissance coalitions solves the challenge of information isolation among shards, while the stateless ordering in the sequence shard and the inter-shard proof sharing in the final correction phase deal with Byzantine failures. 
A rigorous theoretical analysis is provided in \autoref{sec:analysis}. 

\section{Design Refinement}
\label{sec:refinement}


\subsection{Parallelization of Sequencing and Execution}
\label{sec:parallelization}

\textbf{Problem of additional consensus.} \textsc{Prophet} introduces an additional sequence phase in each round. 
This phase requires an intra-shard consensus in the sequence shard, doubling the consensus time for each block.

\textbf{Design.} To solve the problem, \textsc{Prophet} parallelizes the sequence phase and execution phase. Specifically, the leader of the sequence shard can send a global order to the shards before the sequence shard validates the new order through consensus. 
Then, in the execution phase, the sequence shard validates the new global order and pre-execution results. 
An invalid order results in an invalid proof generated by the sequence shard. 
Thus, an invalid order proposed by a malicious leader of the sequence shard can be detected in the correction phase.

\subsection{Fine-grained Ordering}
\label{sec:fine_grained_ordering}

\textbf{Problem of coarse-grained ordering.} In the above system, the reconnaissance coalitions simply define the read/write sets of transactions as the addresses of smart contracts in the pre-execution phase. 
Then, in the sequence phase, the sequence shard only pack the transactions that are not related to the same contracts. 
In such a coarse-grained manner, the transactions accessing the same contract are considered conflict and thus cannot be packed in a global order. 
If there are a majority of conflict transactions in the demand, the throughput of the sequence phase may become a bottleneck of \textsc{Prophet}. 

\textbf{Observation.} To illustrate the performance of the coarse-grained ordering in \autoref{sec:deterministic_sharding} in practice, we collect the history of transactions from Nov-25-2019, Feb-17-2020, and Apr-19-2020 in Ethereum. Then, we execute the transactions in the batch of $10, 100, 500, 1000, 3000$ in parallel to simulate the pre-execution phase in \textsc{Prophet} and evaluate their conflict ratio. 
We denote the approach by contract level. 
As shown in \autoref{fig:conflict_in_batch}, the result shows that the conflict ratio increases with the batch size and nearly 90\% transactions are in conflict when the reconnaissance coalitions pre-execute 3000 transactions.

\begin{figure}[t]
    \centering
	\subfloat[][]{
        \begin{minipage}[t]{0.325\linewidth}
            \centering
            \includegraphics[width=\linewidth]{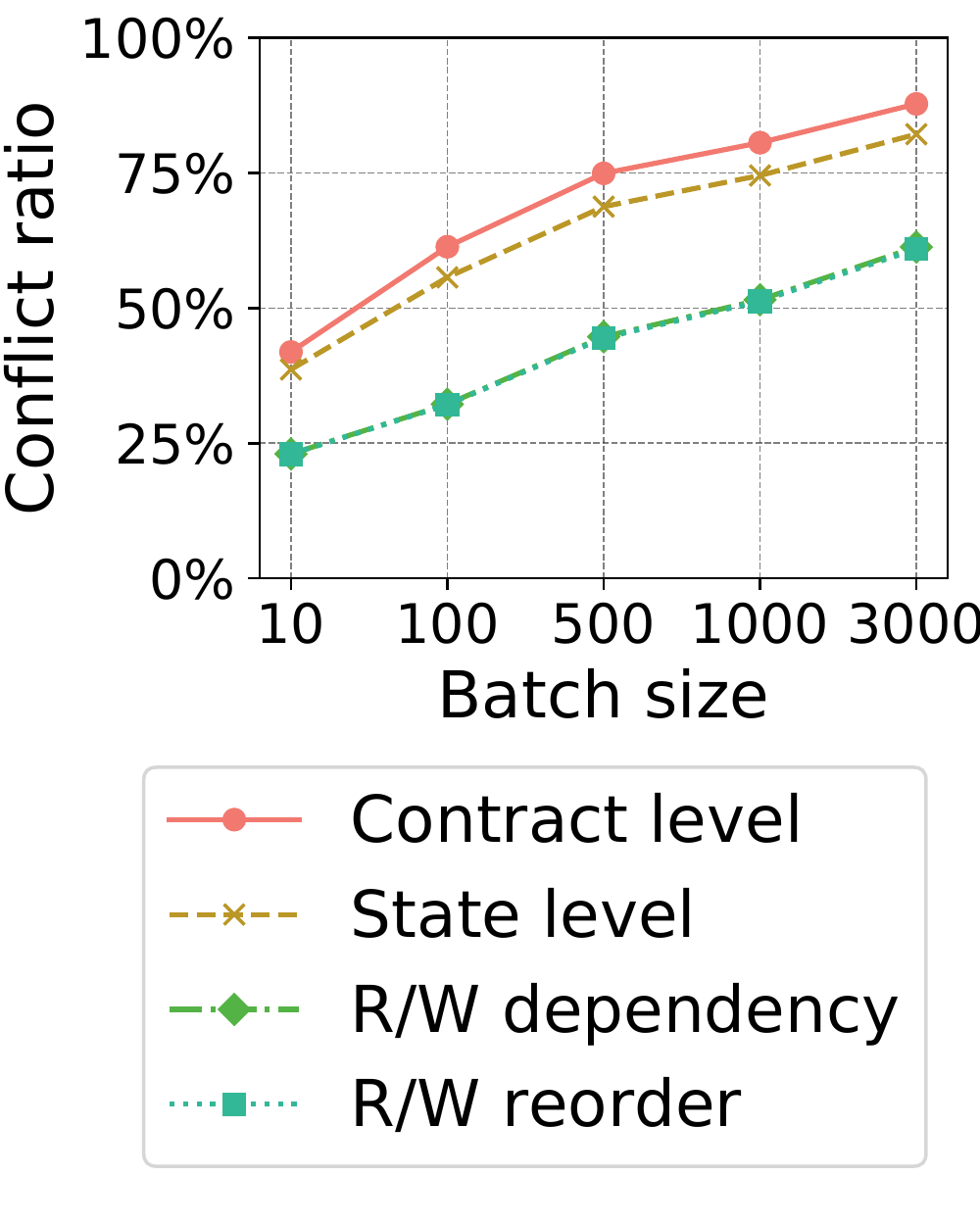}
            \label{fig:conflict_in_batch}
        \vspace{-18pt}
        \end{minipage}
    }
	\subfloat[][]{
        \begin{minipage}[t]{0.615\linewidth}
            \centering
            \includegraphics[width=\linewidth]{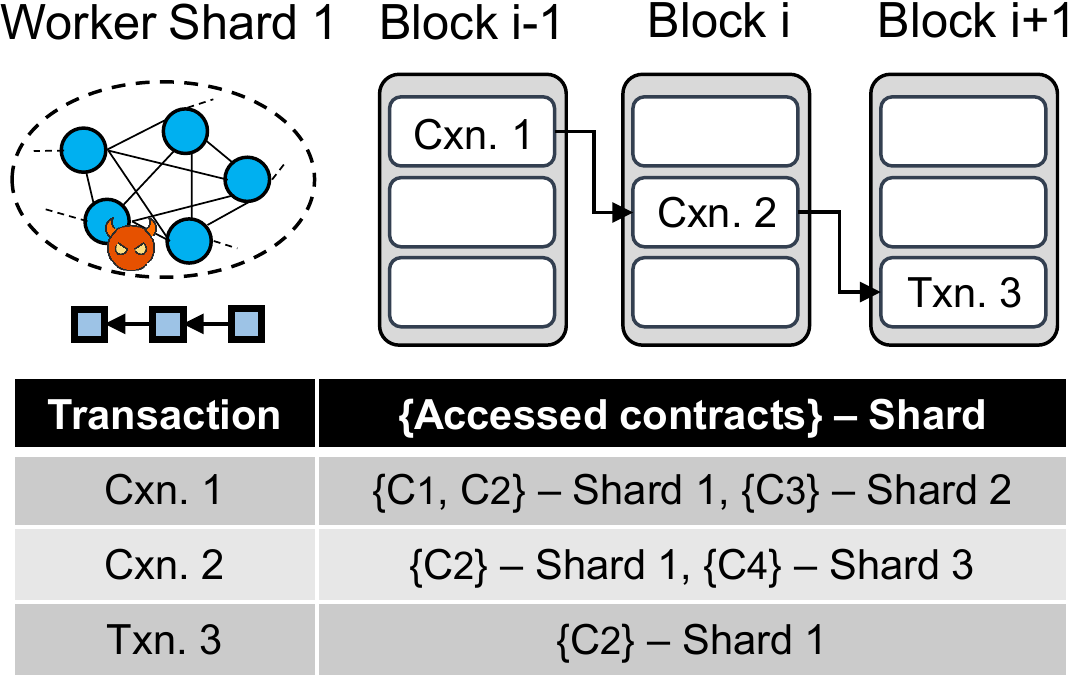}
            \label{fig:asynchronous}
        \vspace{-18pt}
        \end{minipage}
    }
    \caption{(a) Conflict ratio of transactions in a batch with varying batch size in the pre-execution phase; (b) Transaction confirmation rule in asynchronous correction.}
\vspace{-16pt}
\end{figure}


\textbf{Design.} Therefore, we propose a fine-grained read/write ordering approach that is composed of two following steps. 

1) First, we design a fine-grained read/write set identification. We first redefine the read/write sets of a transaction as the blockchain storage that it reads or writes. 
Then, the transactions related to the same contract can access different positions of its storage. 
For example, if two transactions access the same contract but read or write different variables of the contract, they are not in conflict. 
For such a fine-grained identification, we take a deep dive into smart contracts~\cite{ethereum}. 
All contract fields and mappings are saved in blockchain storage and each transaction is a sequence of instructions among which \texttt{SSLOAD} and \texttt{SSTORE} are the two instructions for blockchain persistent storage read and write, respectively.
Thus, during the pre-execution, the reconnaissance coalitions record the instructions \texttt{SSLOAD} and \texttt{SSTORE} and their corresponding addresses. 
As shown in \autoref{fig:conflict_in_batch}, the approach (denoted by state level) reduces the conflict ratio by about 5.34\% compared with the contract-level approach.

2) Second, we design an ordering rule considering the read/write dependency for the sequence phase. If the sequence shard decides an order in which the execution result of any transaction will not influence the read/write sets of its following transactions, the execution of transactions
in the execution phase will be the same as those in the pre-execution phase thus there are no conflicts in the order.
To achieve it, we allow a transaction to have a read-after-read or write-after-read dependency with its previous transactions in the order.
As shown in \autoref{fig:conflict_in_batch}, the approach (denoted by R/W dependency) reduces the conflict ratio by about 21.42\% compared with the state level approach. 
Besides, we also evaluate a reorder rule~\cite{aria}. In detail, for two transactions with read-after-write dependency, the rule can change their position if the new order does not violate the before ordering rule. However, it only reduces the conflict ratio by 0.21\%.




\subsection{Asynchronous Correction}
\label{sec:aynchronous_correction}

\textbf{Problem of synchronous correction.} In the correction phase, for a shard, the validity of its related transactions can be proved and the state of its stored contracts can be updated only when the shard receives all proofs from the other shards. However, in the practice, for each round, there is great uncertainty about both the consensus latency in each shard~\cite{gpbft,bobtail} and the latency of cross-shard communication. This can result in an barrel effect, which means the round time of each shard in \textsc{Prophet} will depend on the slowest shard. 

\textbf{Observation.} The consensus latency has high variance for Proof-of-Work (PoW) protocols adopted by Bitcoin and Ethereum or Byzantine fault tolerance (BFT) protocols adopted by Hyperledger Fabric. For example, although Bitcoin theoretically produces one block every 10 minutes, for 5\% of the time, Bitcoin’s inter-block time is at least 30 minutes~\cite{bobtail}. For PBFT, the consensus latency is uncertain because the state of network environment is often dynamic and elusive~\cite{gpbft}. 

\textbf{Design.} To overcome the problem, \textsc{Prophet} adopts an asynchronous correction design. 
Specifically, in the execution phase, before receiving the proof from the other shards, a shard can optimistically assume that all the transactions are valid. 
Next, it can update the state of its stored contracts based on the current block and move to the next round. 
When a shard receives an invalid proof for a previous transaction from another shard, this previous transaction will be invalidated. 
All the following transactions related to the contracts involved by the invalid transaction will be also invalidated. 
In other words, \textsc{Prophet} has the following confirmation rule for transactions.
\begin{newrule}
A transaction $\mathcal{T}$ can be confirmed by a shard only when the shard receives all the related proofs of $\mathcal{T}$ and all the related proofs of the previous transactions related to $\mathcal{T}$. 
\end{newrule}

For example, \autoref{fig:asynchronous} shows three blocks, i.e., Block i-1, i, and i+1, of Shard 1. In these three blocks, there are three transactions, i.e., Cxn.~1, Cxn.~2 and Txn.~3, all of which access the same contract $C_2$ stored in Shard 1. 
Based on the confirmation rule of transactions, the confirmation of Txn.~3 depends on the proof of Shard 2 in Block i-1, the proof of Shard 3 in Block i, and the proof of Shard 1 in Block i+1.

\subsection{Parallel Pre-Execution}
\label{sec:communication}

\textbf{Problem of serial pre-execution.} The throughput of \textsc{Prophet} also depends on the total pre-executed throughput of reconnaissance coalitions. The above system considers a serial pre-execution approach in which the nodes in each reconnaissance coalition execute transactions one by one. However, when the communication accounts for a higher portion than the contract execution as proved below, each node may spend a lot of time on the communication of cross-shard inter-contract calls, keeping its CPU idle most of the time and restricting the pre-executed transaction throughput. 

\textbf{Observation}. To find the main bottleneck of pre-executed throughput, we evaluate the simplest cooperation mode for a reconnaissance coalition as shown in \autoref{fig:parallel_cooperation_1}. Specifically, the nodes in the reconnaissance coalition executes transaction one by one.
Based on the transactions collected in \autoref{sec:fine_grained_ordering}, the communication time accounts for 87.5\% of the total time, since most contracts' computation is simple. 

\textbf{Design}. To minimize the communication overhead, we propose an overlap cooperation mode that overlaps the computation process and communication process during pre-execution. 
For example, as shown in \autoref{fig:parallel_cooperation_2}, after meeting the first cross-shard contract call in Cxn.~1, Node $N_1$ can transmit a cross-shard message to Node $N_2$ while simultaneously executing the computation task of Cxn.~2. 
Through this way, the computing resource and communication resource could achieve nearly full utilization. 
Furthermore, since all transactions are pre-executed based on the state of the previous block, we also propose a parallel cooperation mode.
In particular, each node executes different transactions at the same time using redundant computation resources. For example, as shown in \autoref{fig:parallel_cooperation_3}, Node $N_1$ executes Cxn.~1 and 2 using two threads, i.e., Thread 1 and 2, respectively.
As evaluated in \autoref{sec:micro_benchmark}, the parallel scheme can increase the pre-execution throughput by $86.9\%\sim280.0\%$ compared with the sequential scheme.

\section{Analysis}
\label{sec:analysis}

We first show how \textsc{Prophet} achieves both determinism and serializability.
The former one means that the same result is always produced in all honest node for each shard. 
The later one requires transactions in the system to produce the results following some serial order. 
The analysis depends on the intra-shard consensus of shards in \textsc{Prophet} thus we define $v$ as the fault threshold of the adopted intra-shard consensus~\cite{multi_threshold}. 
For example, the synchronous protocol in Rapidchain~\cite{rapidchain} tolerate up to $v=1/2$ Byzantine faults, while the asynchronous or partially synchronous protocol in Omniledger~\cite{omniledger} tolerates only up to $v=1/3$ Byzantine faults.

\begin{figure}[t]
	\centering
	\subfloat[][Sequential]{
		\begin{minipage}[t]{0.51\linewidth}
			\centering
			\includegraphics[width=\linewidth]{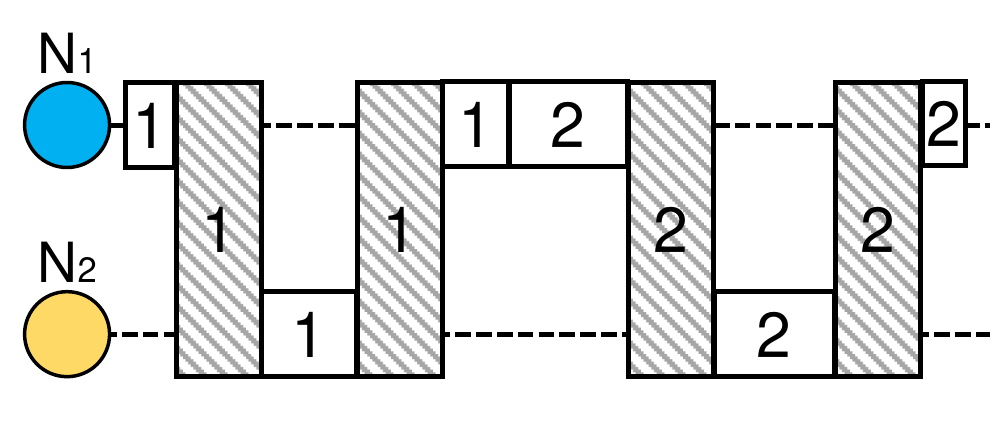}
		\end{minipage}
		\label{fig:parallel_cooperation_1}
	}
	\subfloat[][Overlap]{
		\begin{minipage}[t]{0.32\linewidth}
			\centering
			\includegraphics[width=\linewidth]{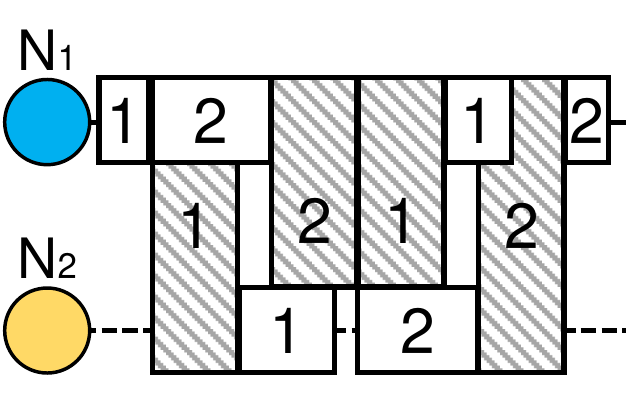}
		\end{minipage}%
		\label{fig:parallel_cooperation_2}
	}
	\\\vspace{-5pt}
	\subfloat[][Parallel]{
		\begin{minipage}[t]{0.41\linewidth}
			\centering
			\includegraphics[width=\linewidth]{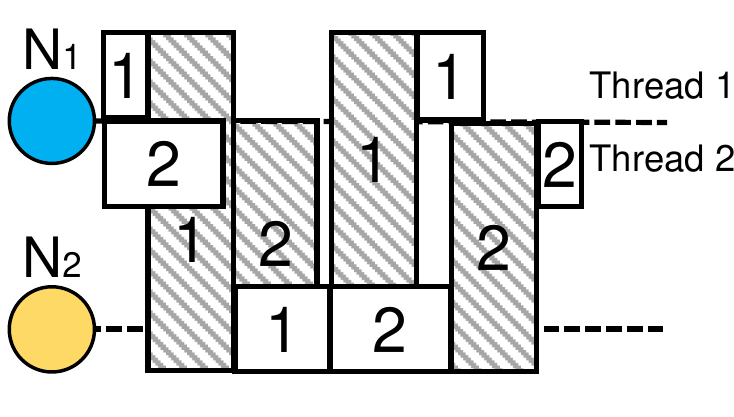}
		\end{minipage}
		\label{fig:parallel_cooperation_3}
	}
	\subfloat[][Legend]{
		\begin{minipage}[t]{0.4\linewidth}
			\centering
			\includegraphics[width=\linewidth]{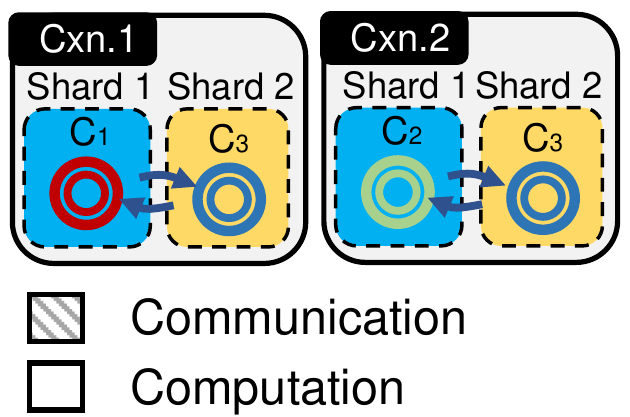}
		\end{minipage}%
		\label{fig:parallel_cooperation_4}
	}
	\caption{Comparison of three cooperation modes for a reconnaissance coalition in the pre-execution phase. The number inside each rectangle denotes the transaction ID to which the computation or communication time belongs.}
	\vspace{-16pt}
\end{figure}

\begin{theorem}
\textsc{Prophet} achieves determinism and serializability if there are no more than $v$ fraction of malicious nodes in each shard.
\label{thm:thm1}
\end{theorem}

\begin{proof}
When there are no more than $v < \frac{1}{3}$ malicious nodes in each shard, the intra-shard consensus can guarantee safety~\cite{rapidchain,sigmod_sharding,pyramid}, i.e., the honest nodes in each shard agree on the same valid block in each round.
Thus, the intra-shard consensus can guarantee that both the order proposed by the sequence shard follows the rule in \autoref{sec:sequence} and \autoref{sec:fine_grained_ordering} and the validation proofs proposed by the shards are valid. 
It also guarantees that a message along with a collective signature is honest because malicious nodes are the minority, i.e., no more than $v$, of the shard. 
Moreover, the message of each party (e.g., transaction profile, order, and proof) cannot be modified and forged since the collective signature can be used to detect forgery or tampering. 
Finally, because the correction phase guarantees that any invalid transaction will not be committed in all its related shards, all the honest nodes in the sharding system can run an identical batch of transactions based on the same global serial order and the same blockchain state. 
Additionally, the code of smart contracts is deterministic~\cite{ethereum}, which means each node can get the same result given the same input for a contract method.
It guarantees the determinism of the consensus in \textsc{Prophet}.

Next, we prove the serializability by contradiction as follows. 
Assume the global order produced by the sequence shard is: $\cdots \rightarrow T_i \rightarrow \cdots \rightarrow T_j \rightarrow \cdots$ where $T_i$ and $T_j$ can be two transactions in the same shard or in the different shards. There are two possible outcomes to violate the serializability. The first one is that $T_j$'s update is overwritten by $T_i$'s update. The second one is that $T_i$ reads $T_j$'s update. However, for \autoref{sec:sequence}, $T_i$ and $T_j$ will not access the same contract. And, for \autoref{sec:fine_grained_ordering}, the sequence shard only allows read-after-read and write-after-read dependency, thus both outcomes result in a contradiction and the consensus in \textsc{Prophet} achieves serializability.
In addition, we emphasize that even if $T_i$ is invalidated in the correction phase, the following transactions in the global order will not be influenced. 
It is because $T_i$ is not allowed to change the state of contracts that are read or written by its following transactions considering the read-after-write is not allowed.
\end{proof}

Similar to the other sharding systems~\cite{elastico, rapidchain, omniledger, pyramid}, \textsc{Prophet} has a global fault threshold for the whole sharding system denoted by $f$ and a security parameter denoted by $\lambda$. After dividing each node to a random shard, the proportion of malicious nodes in each shard for \textsc{Prophet} can be proven to be lower than the fault threshold $v$ with low probability, i.e., the probability is no more than $2^{-\lambda}$, thus \autoref{thm:thm1} can be guaranteed with high probability in \textsc{Prophet}.

\textbf{Overhead Analysis.} In terms of the time overhead, \textsc{Prophet} parallelizes the sequence phase and execution phase in \autoref{sec:parallelization}; thus, there is one consensus in every shard for each round, similar to the existing blockchain sharding. 
Besides, asynchronous correction in \autoref{sec:aynchronous_correction} enables each shard to move to the next round without waiting for the other shards' proof after the execution phase, saving the time of the correction phase. 
Therefore, only the pre-execution phase introduces an additional time overhead for each round. 
In terms of the computation overhead, the pre-execution phase introduces some additional computation tasks to reconnaissance coalitions.

\section{Implementation}

We implement a prototype of \textsc{Prophet} based on Geth~\cite{geth}, the Go language implementation of Ethereum. The smart contracts in \textsc{Prophet} run in EVM in Geth. We adopt a BFT consensus with BLS multi-signature~\cite{harmony_consensus} as the intra-shard consensus of \textsc{Prophet}. For comparison, we also implement two non-deterministic sharding prototypes. 
Since the intra-shard consensus in \textsc{Prophet} can be substituted by any other BFT consensus, to ensure the result will not be affected by the difference in intra-shard consensus, we adopt the same consensus for the intra-shard consensus in these two non-deterministic sharding prototypes. 
Moreover, for a fair comparison, both prototypes are equipped with the fine-grained read/write approach in \autoref{sec:fine_grained_ordering}.
The difference between two prototypes is the cross-shard transaction processing. The first one uses the OCC mechanism in Monoxide~\cite{monoxide} and the second one uses the 2PL mechanism in Chainspace~\cite{chainspace}. Their main ideas are referred to in \autoref{sec:background}.

\section{Evaluation}
\label{sec:evaluation}

\noindent
\textbf{Dataset.} To evaluate our sharding system \textsc{Prophet} on the historical transactions in Ethereum, we implement a smart contract recorder/replayer based on EVM stateless state transition tool~\cite{evm} similar to~\cite{273865} and collect the blocks from Nov-25-2019 to May-04-2020 (block height: 9,000,000-10,000,000) from Ethereum mainnet blockchain. 

\noindent
\textbf{Setup.} The number of nodes in each shard is set as $50$. In OCC and 2PL, the maximum retry count for the transactions is set as $10$, which means that a transaction with more than 10 retries will be aborted. 
The testbed is composed of 16 machines, each of which has an Intel E5-2680V4 CPU and 64 GB of RAM, and a 10 Gbps network link. Similar to~\cite{rapidchain, omniledger}, to simulate geographically-distributed nodes, we set the bandwidth of all connections between nodes to $20$ Mbps and impose a latency of $100$ ms on the links in our testbed. 
The proportion of malicious nodes in the system is set as $12.5\%$.
In our setting, the malicious nodes in each reconnaissance coalition provide invalid cross-shard messages and read/write sets to interrupt the pre-execution phase.
We repeat each experiment three times and compute the average as its result.

\noindent
\textbf{Metrics.} We measure the performance of a sharding system using the following metrics. 1) \textit{Transaction throughput}: the throughput of the confirmed transactions measured in TPS. 2) \textit{Confirmation latency}: the delay between the time that a transaction is issued by a client until it can be confirmed by any (honest) node in the system. 3) \textit{Abort ratio}: the ratio of aborted transactions during commitment, i.e., the transaction whose retry count exceeds the maximum retry count as discussed above. 4) \textit{Invalid ratio}: the ratio of invalid transactions found in the correction phase for \textsc{Prophet}.

\subsection{Performance}

\begin{figure}[t]
    \centering
    \includegraphics[width=0.9\linewidth]{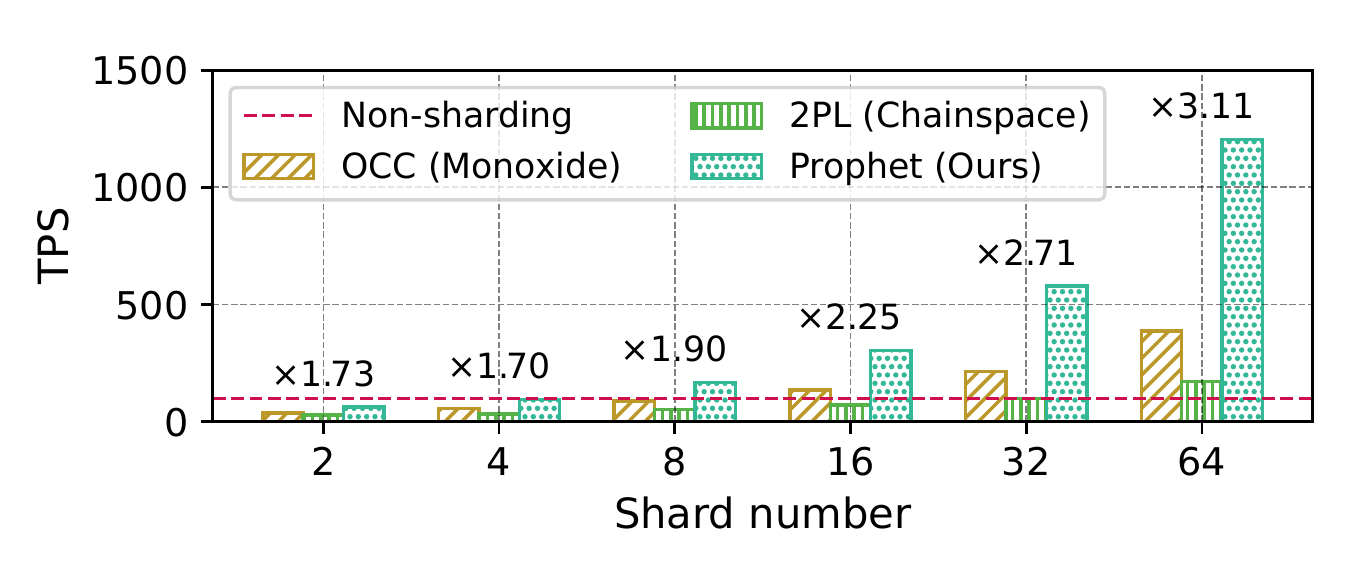}
    \caption{Transaction throughput of \textsc{Prophet} and the existing sharding systems (The number above each bar denotes the ratio of the throughput of \textsc{Prophet} over that of OCC.)}
    \label{fig:tps}
\vspace{-13pt}
\end{figure}

\begin{figure}[t]
    \centering
    \includegraphics[width=0.9\linewidth]{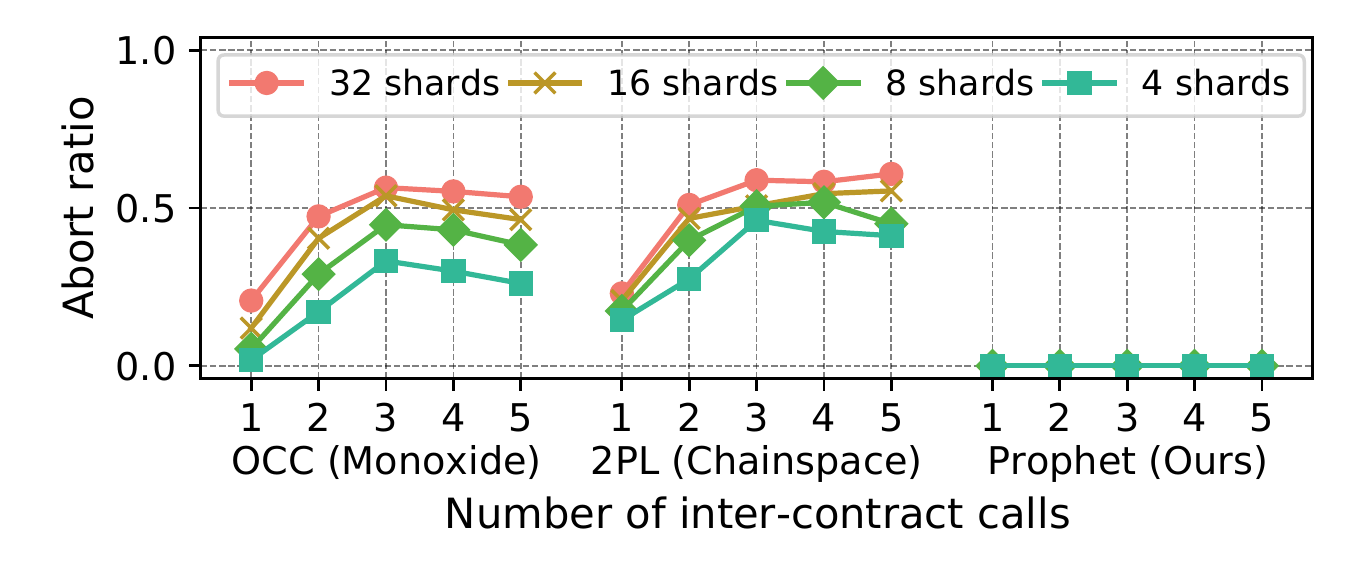}
    \caption{Abort ratio of transactions during commitment in \textsc{Prophet} and the existing sharding works.}
    \label{fig:conflict}
\vspace{-16pt}
\end{figure}

\begin{figure*}[t]
    \centering
    \begin{minipage}[t]{0.32\textwidth}
        \centering
        \includegraphics[width=1.05\linewidth]{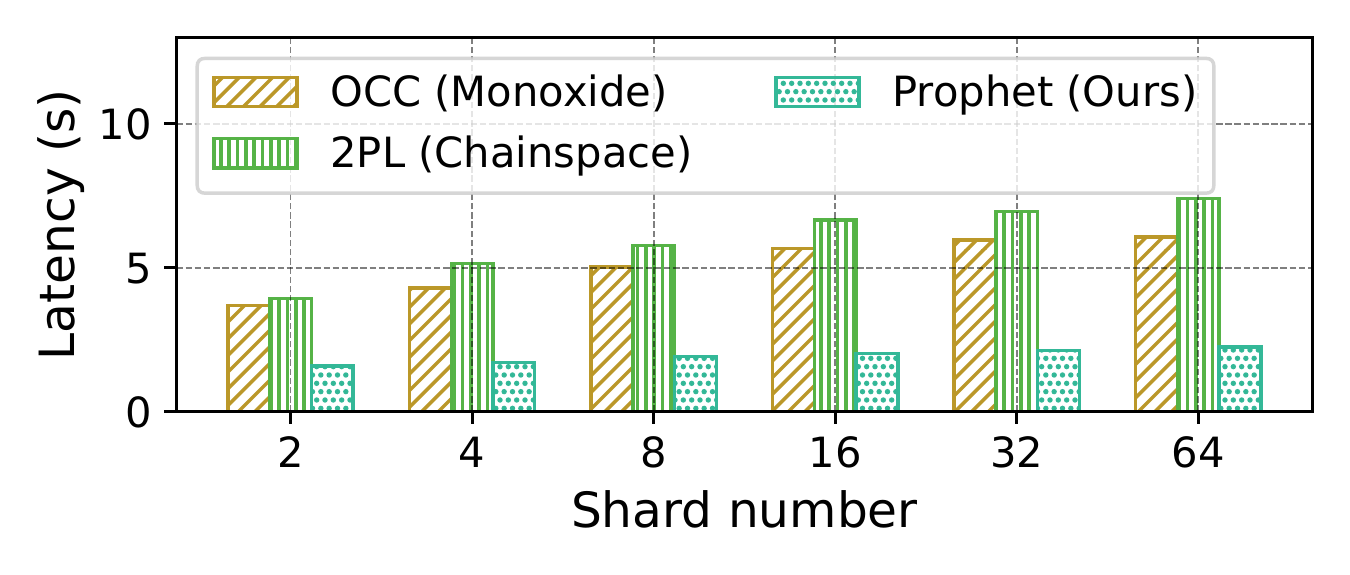}
        \caption{Confirmation latency of \textsc{Prophet} and the other blockchain sharding systems.}
        \label{fig:latency}
    \end{minipage}
    \hfill
    \begin{minipage}[t]{0.32\textwidth}
        \centering
        \includegraphics[width=1.05\linewidth]{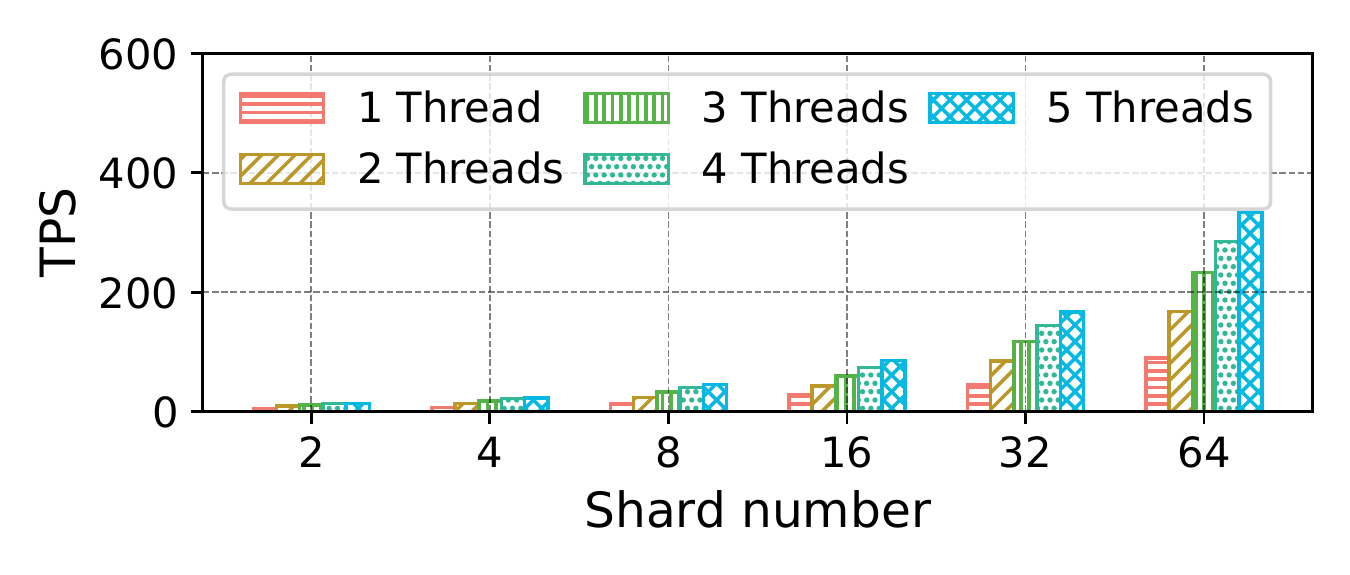}
        \caption{Pre-execution throughput of a reconnaissance coalition for \textsc{Prophet} with different number of shards.}
        \label{fig:pre_execution}
    \end{minipage}
    \hfill
    \begin{minipage}[t]{0.32\textwidth}
        \centering
        \includegraphics[width=1.05\linewidth]{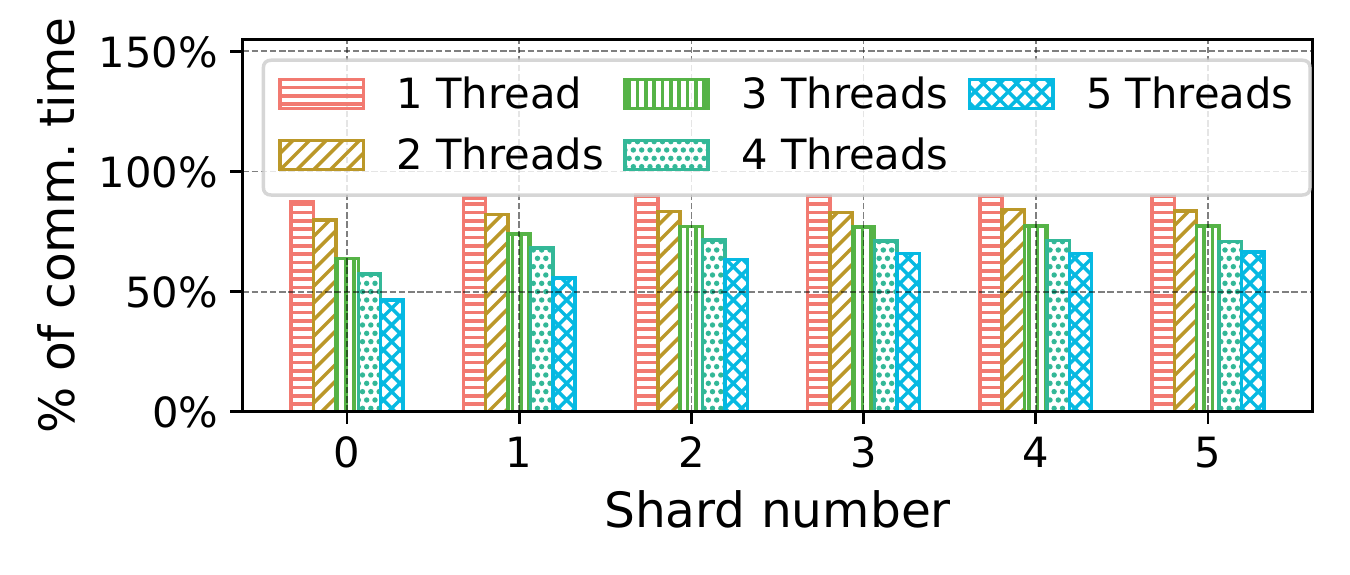}
        \caption{Percentage of communication time in the pre-execution in \textsc{Prophet}.}
        \label{fig:pre_execution_communication}
    \end{minipage}
\vspace{-15pt}
\end{figure*}

To evaluate the performance, we measure the throughput in TPS for the two non-deterministic sharding blockchains and \textsc{Prophet} with varying number of shards. As shown in \autoref{fig:tps}, all three sharding systems achieve the linear scalability. However, \textsc{Prophet} improves the throughput $1.73\sim3.11$X against the two traditional sharding systems and the improvement is more significant when there are more shards. 
Moreover, the throughput of the sharding systems is even worse than that of non-sharding system when there are less shards, the reason of which is twofold. 
First, introducing sharding into the blockchain results in cross-shard transactions, each of which needs to be divided into multiple sub-transactions and processed in multiple consensus rounds. 
Second, as discussed in \autoref{sec:background}, the contention of cross-shard transactions results in frequent aborts of transactions thus most of the throughput in the two traditional sharding systems is wasted.

To investigate the wasted throughput reason as described above, we measure the abort ratio of transactions with varying number of inter-contract calls for the two non-deterministic sharding blockchains and \textsc{Prophet}.
Note that although there are many transactions with more than 4 inter-contract calls in the system, \autoref{fig:conflict} only shows the transactions with 1-5 inter-contract calls due to the space constraint.
\autoref{fig:conflict} shows that the abort ratio in both of all sharding systems increases as the number of shards increases. 
Specifically, the OCC-based system aborts about 50\% transactions with more than 2 inter-contract calls. 
It limits the dApps consisting of complex smart contract interactions. 
In comparison, \textsc{Prophet} keeps the abort ratio being 0 no matter how many cross-shard contract calls the transactions include.


We then evaluate the confirmation latency of transactions in \textsc{Prophet} and the non-deterministic sharding blockchains with varying number of shards. \autoref{fig:latency} shows that the latency increases as the number of shards increases in the non-deterministic systems. 
This is because the increase of shards can introduce more cross-shard transactions that need more consensus round to be committed, thus the confirmation latency is higher. In comparison, the latency in \textsc{Prophet} is low and relatively stable since any cross-shard transaction can be committed by \textsc{Prophet} in one round.


\subsection{Micro-benchmark}
\label{sec:micro_benchmark}


To analyze the effectiveness of our communication efficient pre-execution proposed in \autoref{sec:communication}, we evaluate the pre-execution throughput of a reconnaissance coalition with varying number of shards. As shown in \autoref{fig:pre_execution}, the pre-execution throughput increases with the number of threads in each node. When there are more shards in \textsc{Prophet}, the number of nodes in a reconnaissance coalition increases and each node has a thread. This parallelism improvement outweighs the increase of cross-shard transactions. Furthermore, a reconnaissance coalition achieves 334 TPS with 5 threads. Based on combining this observation and \autoref{fig:conflict_in_batch}, we can get that when there are more than 10 reconnaissance coalitions, the pre-execution throughput can exceed the maximum throughput (i.e., 1203 TPS in \autoref{fig:tps}) in \textsc{Prophet}. Therefore, the bottleneck is not in the pre-execution phase.

To further investigate the improvement of communication efficient pre-execution, we evaluate the proportion of communication time in the total time in the pre-execution phase and the result is shown in \autoref{fig:pre_execution_communication}. Note that when the communication is overlapped by the computation, we do not consider the communication time in the total time. From the figure, we can see that the proportion of communication time is less in a node with more threads.


We also evaluate the average total size of cross-shard messages for a transaction in a system with varying number of shards. \autoref{table:detail_query_time} shows that the average total message size increases with the number of shards. It is because the number of shards can result in more cross-shard contract calls in a transaction. Moreover, the increment of message size gradually decreases. Specifically, doubling the shard number from 32 to 64 only increases the message size by 3 Bytes. It is because the number of cross-shard contract calls will be mainly influenced by the number of contract calls when the shard number is more than the number of contract calls.

\begin{table}[t]
	\centering
	\caption{The average total size of cross-shard messages for a transaction in a system with varying number of shards.}
	\scalebox{1}{
    \begin{tabular}{c|c|c|c|c|c|c}
    \hline 
    Shard number & 2 & 4 & 8 & 16 & 32 & 64 \\
    \hline\hline
    Message size (Byte) & 177 & 240 & 267 & 310 & 319 & 322 \\
    \hline
    \end{tabular}
    }
	\label{table:detail_query_time}
    \vspace{-10pt}
\end{table}

Although malicious nodes cannot make the invalid transactions be confirmed because of the correction phase, they can occupy the throughput for the valid transactions.
We evaluate the invalid ratio of \textsc{Prophet} with different percentage of malicious nodes (in particular 1\%, 5\%, 12.5\%, and 25\%) and the result is illustrated in \autoref{fig:malicious_with_incentive}.
The malicious nodes have only a limited impact (less than 2\%) on the throughput of \textsc{Prophet}.
\autoref{fig:invalid_fluctuation} shows that the invalid ratio decreases over time because of the gradual construction of honest coalitions.
As discussed in \autoref{sec:pre_execution}, to gain more transaction fees, the honest nodes tend to form and stay in reconnaissance coalitions involving the other honest nodes and leave coalitions involving malicious nodes.

\begin{figure}[t]
    \centering
	\subfloat[][Average invalid ratio.]{
        \begin{minipage}[t]{0.36\linewidth}
            \centering
            \includegraphics[width=\linewidth]{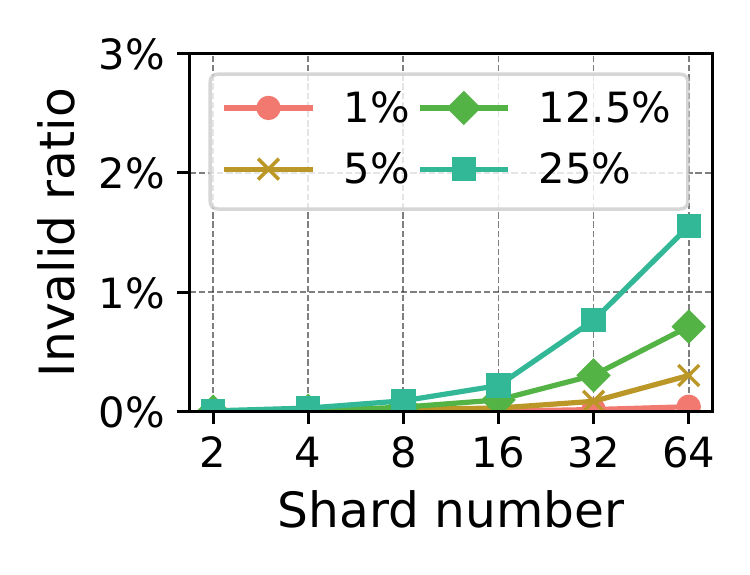}
        \end{minipage}
        \label{fig:malicious_with_incentive}
    }
    \subfloat[][Fluctuation of invalid ratio for 64 shards.]{
        \begin{minipage}[t]{0.63\linewidth}
            \centering
            \includegraphics[width=\linewidth]{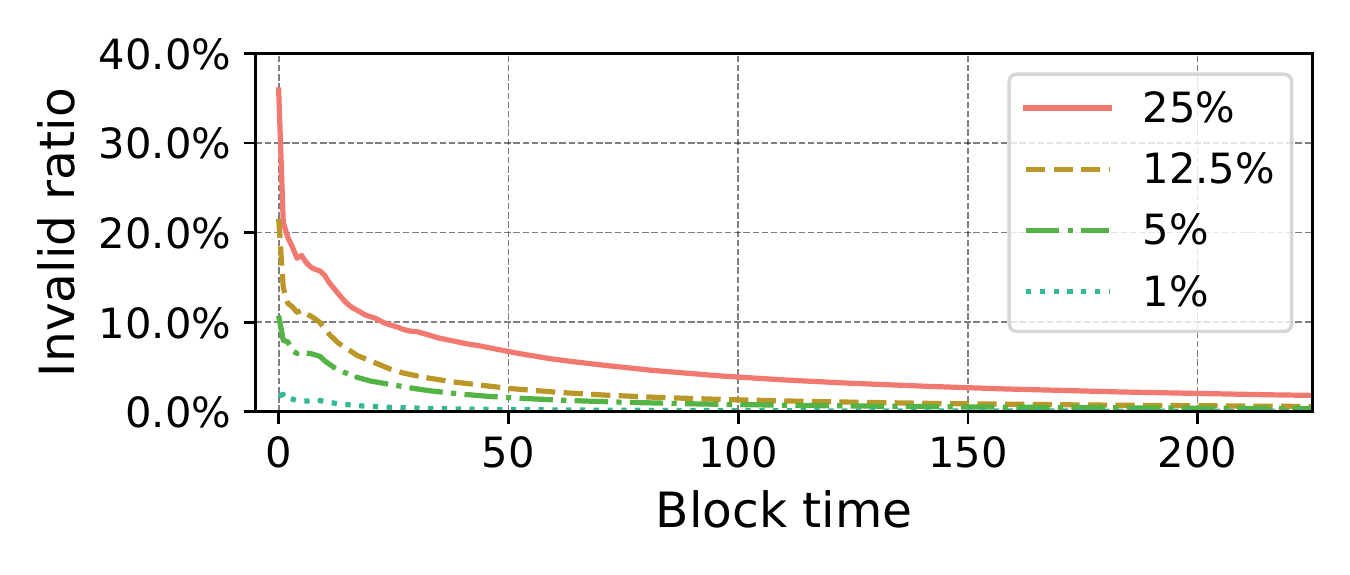}
        \end{minipage}
        \label{fig:invalid_fluctuation}
    }
    \centering
	\caption{Ratio of invalid transactions in the correction phase.}
    \vspace{-15pt}
\end{figure}

\section{Conclusion}

We present \textsc{Prophet}, a sharding blockchain for conflict-free transactions. \textsc{Prophet} achieves conflict-free by introducing a layer-2 sharding architecture on top of the shards of the existing blockchain sharding. The running of the architecture depends on the cooperation and supervision among reconnaissance coalitions, sequence shard, and shards. \textsc{Prophet} also features several improved designs for ordering efficiency, such as fine-grained ordering, asynchronous correction, and parallel execution. Experimental evaluations show that \textsc{Prophet} 
boosts the throughput of $3.11\times$ (i.e., 1203 TPS), decreases the latency by $62.9$\%, and achieves nearly no aborts compared with the existing blockchain sharding systems. In the future work, we will study commutative contract instructions for finer-grained ordering and fairness-aware incentive.



\bibliographystyle{IEEEtran}
\bibliography{sample}

\begin{thebibliography}{10}
\providecommand{\url}[1]{#1}
\csname url@samestyle\endcsname
\providecommand{\newblock}{\relax}
\providecommand{\bibinfo}[2]{#2}
\providecommand{\BIBentrySTDinterwordspacing}{\spaceskip=0pt\relax}
\providecommand{\BIBentryALTinterwordstretchfactor}{4}
\providecommand{\BIBentryALTinterwordspacing}{\spaceskip=\fontdimen2\font plus
\BIBentryALTinterwordstretchfactor\fontdimen3\font minus
  \fontdimen4\font\relax}
\providecommand{\BIBforeignlanguage}[2]{{%
\expandafter\ifx\csname l@#1\endcsname\relax
\typeout{** WARNING: IEEEtran.bst: No hyphenation pattern has been}%
\typeout{** loaded for the language `#1'. Using the pattern for}%
\typeout{** the default language instead.}%
\else
\language=\csname l@#1\endcsname
\fi
#2}}
\providecommand{\BIBdecl}{\relax}
\BIBdecl

\bibitem{werner2021sok}
S.~M. Werner, D.~Perez, L.~Gudgeon, A.~Klages-Mundt, D.~Harz, and W.~J.
  Knottenbelt, ``Sok: Decentralized finance (defi),'' 2021.

\bibitem{wang2021nonfungible}
Q.~Wang, R.~Li, Q.~Wang, and S.~Chen, ``Non-fungible token (nft): Overview,
  evaluation, opportunities and challenges,'' 2021.

\bibitem{triple}
CoinDesk, ``Soaring defi usage drives ethereum contract calls to new record,''
  \url{https://www.coindesk.com/soaring-defi-usage-drives-ethereum-contract-calls-to-new-record}.

\bibitem{bitcoin}
S.~Nakamoto, ``Bitcoin: A peer-to-peer electronic cash system,'' 2008.

\bibitem{ethereum}
G.~Wood, ``Ethereum: A secure decentralised generalised transaction ledger,''
  \emph{Ethereum project yellow paper}, 2014.

\bibitem{sok_sharding}
G.~Wang, Z.~J. Shi, M.~Nixon, and S.~Han, ``Sok: Sharding on blockchain,'' in
  \emph{Proceedings of the 1st ACM Conference on Advances in Financial
  Technologies (AFT)}, 2019, pp. 41--61.

\bibitem{elastico}
L.~Luu, V.~Narayanan, C.~Zheng, K.~Baweja, S.~Gilbert, and P.~Saxena, ``A
  secure sharding protocol for open blockchains,'' in \emph{Proceedings of the
  2016 ACM SIGSAC Conference on Computer and Communications Security (CCS)},
  2016, pp. 17--30.

\bibitem{omniledger}
E.~{Kokoris-Kogias}, P.~{Jovanovic}, L.~{Gasser}, N.~{Gailly}, E.~{Syta}, and
  B.~{Ford}, ``Omniledger: A secure, scale-out, decentralized ledger via
  sharding,'' in \emph{IEEE Symposium on Security and Privacy (S\&P)}, 2018,
  pp. 583--598.

\bibitem{rapidchain}
M.~Zamani, M.~Movahedi, and M.~Raykova, ``Rapidchain: Scaling blockchain via
  full sharding,'' in \emph{Proceedings of the 2018 ACM SIGSAC Conference on
  Computer and Communications Security (CCS)}, 2018, pp. 931--948.

\bibitem{monoxide}
J.~Wang and H.~Wang, ``Monoxide: Scale out blockchains with asynchronous
  consensus zones,'' in \emph{USENIX Symposium on Networked Systems Design and
  Implementation (NSDI)}, 2019, pp. 95--112.

\bibitem{pyramid}
Z.~Hong, S.~Guo, P.~Li, and W.~Chen, ``Pyramid: A layered sharding blockchain
  system,'' in \emph{IEEE International Conference on Computer Communications
  (INFOCOM)}, 2021, pp. 1--10.

\bibitem{pldi}
G.~P\^{\i}rlea, A.~Kumar, and I.~Sergey, ``Practical smart contract sharding
  with ownership and commutativity analysis,'' in \emph{Proceedings of the 42nd
  ACM SIGPLAN International Conference on Programming Language Design and
  Implementation (PLDI)}, 2021, pp. 1327--1341.

\bibitem{byshard}
J.~Hellings and M.~Sadoghi, ``Byshard: Sharding in a byzantine environment,''
  \emph{Proc. VLDB Endow.}, vol.~14, no.~11, pp. 2230--2243, jul 2021.

\bibitem{brokerchain}
H.~Huang, X.~Peng, J.~Zhan, S.~Zhang, Y.~Lin, Z.~Zheng, and S.~Guo,
  ``Brokerchain: A cross-shard blockchain protocol for account/balance-based
  state sharding,'' in \emph{IEEE International Conference on Computer
  Communications (INFOCOM)}, 2022, pp. 1968--1977.

\bibitem{sstore}
X.~Qi, ``S-store: A scalable data store towards permissioned blockchain
  sharding,'' in \emph{IEEE International Conference on Computer Communications
  (INFOCOM)}, 2022, pp. 1978--1987.

\bibitem{zilliqa}
Z.~team, ``Zilliqa,'' \url{https://www.zilliqa.com/}.

\bibitem{layer2}
Ethereum, ``Shard chains,'' \url{https://ethereum.org/en/eth2/shard-chains/}.

\bibitem{sigmod_sharding}
H.~Dang, T.~T.~A. Dinh, D.~Loghin, E.-C. Chang, Q.~Lin, and B.~C. Ooi,
  ``Towards scaling blockchain systems via sharding,'' in \emph{Proceedings of
  the 2019 International Conference on Management of Data (SIGMOD)}, 2019, pp.
  123--140.

\bibitem{chainspace}
M.~Al{-}Bassam, A.~Sonnino, S.~Bano, D.~Hrycyszyn, and G.~Danezis,
  ``Chainspace: {A} sharded smart contracts platform,'' in \emph{The Network
  and Distributed System Security (NDSS) Symposium}, 2018.

\bibitem{calvin}
A.~Thomson, T.~Diamond, S.-C. Weng, K.~Ren, P.~Shao, and D.~J. Abadi, ``Calvin:
  Fast distributed transactions for partitioned database systems,'' in
  \emph{Proceedings of the 2012 ACM SIGMOD International Conference on
  Management of Data (SIGMOD)}, 2012, pp. 1--12.

\bibitem{aria}
Y.~Lu, X.~Yu, L.~Cao, and S.~Madden, ``Aria: A fast and practical deterministic
  oltp database,'' \emph{Proc. VLDB Endow.}, vol.~13, no.~12, pp. 2047--2060,
  jul 2020.

\bibitem{prescient}
Y.-S. Lin, C.~Tsai, T.-Y. Lin, Y.-S. Chang, and S.-H. Wu, ``Don't look back,
  look into the future: Prescient data partitioning and migration for
  deterministic database systems,'' in \emph{Proceedings of the 2021
  International Conference on Management of Data (SIGMOD/PODS)}, 2021, pp.
  1156--1168.

\bibitem{pwv}
J.~M. Faleiro, D.~J. Abadi, and J.~M. Hellerstein, ``High performance
  transactions via early write visibility,'' \emph{Proc. VLDB Endow.}, vol.~10,
  no.~5, pp. 613--624, jan 2017.

\bibitem{deterministic_overview}
D.~J. Abadi and J.~M. Faleiro, ``An overview of deterministic database
  systems,'' \emph{Communications of the ACM}, vol.~61, no.~9, pp. 78--88, aug
  2018.

\bibitem{icde}
Y.~Tao, B.~Li, J.~Jiang, H.~C. Ng, C.~Wang, and B.~Li, ``On sharding open
  blockchains with smart contracts,'' in \emph{IEEE International Conference on
  Data Engineering (ICDE)}, 2020, pp. 1357--1368.

\bibitem{blockchaindb}
M.~El-Hindi, C.~Binnig, A.~Arasu, D.~Kossmann, and R.~Ramamurthy,
  ``Blockchaindb: A shared database on blockchains,'' \emph{Proc. VLDB Endow.},
  vol.~12, no.~11, pp. 1597--1609, jul 2019.

\bibitem{skychain}
J.~Zhang, Z.~Hong, X.~Qiu, Y.~Zhan, S.~Guo, and W.~Chen, ``Skychain: A deep
  reinforcement learning-empowered dynamic blockchain sharding system,'' in
  \emph{49th International Conference on Parallel Processing (ICPP)}, 2020, pp.
  1--11.

\bibitem{sharper}
M.~J. Amiri, D.~Agrawal, and A.~El~Abbadi, ``Sharper: Sharding permissioned
  blockchains over network clusters,'' in \emph{Proceedings of the 2021
  International Conference on Management of Data (SIGMOD)}, 2021, pp. 76--88.

\bibitem{worldstate}
B.~Research, ``Bitcoin vs ethereum – blockchain size,''
  \url{https://blog.bitmex.com/bitcoin-vs-ethereum-blockchain-size/}.

\bibitem{cosi}
E.~{Syta}, I.~{Tamas}, D.~{Visher}, D.~I. {Wolinsky}, P.~{Jovanovic},
  L.~{Gasser}, N.~{Gailly}, I.~{Khoffi}, and B.~{Ford}, ``Keeping authorities
  "honest or bust" with decentralized witness cosigning,'' in \emph{IEEE
  Symposium on Security and Privacy (S\&P)}, 2016, pp. 526--545.

\bibitem{bls}
D.~Boneh, B.~Lynn, and H.~Shacham, ``Short signatures from the weil pairing,''
  in \emph{Proceedings of the 7th International Conference on the Theory and
  Application of Cryptology and Information Security: Advances in Cryptology
  (ASIACRYPT)}, 2001, pp. 514--532.

\bibitem{gpbft}
L.~Lao, X.~Dai, B.~Xiao, and S.~Guo, ``G-pbft: A location-based and scalable
  consensus protocol for iot-blockchain applications,'' in \emph{IEEE
  International Parallel and Distributed Processing Symposium (IPDPS)}, 2020,
  pp. 664--673.

\bibitem{bobtail}
G.~Bissias and B.~N. Levine, ``Bobtail: Improved blockchain security with
  low-variance mining,'' in \emph{The Network and Distributed System Security
  (NDSS) Symposium}, 2022.

\bibitem{multi_threshold}
A.~Momose and L.~Ren, ``Multi-threshold byzantine fault tolerance,'' in
  \emph{Proceedings of the 2021 ACM SIGSAC Conference on Computer and
  Communications Security (CCS)}, 2021, pp. 1686--1699.

\bibitem{geth}
Ethereum, ``Go ethereum,'' \url{https://github.com/ethereum/go-ethereum}.

\bibitem{harmony_consensus}
Harmony, ``Harmony consensus protocol design,''
  \url{https://github.com/harmony-one/harmony/tree/main/consensus}.

\bibitem{evm}
Ethereum, ``Evm state transition tool,''
  \url{https://github.com/ethereum/go-ethereum/tree/master/cmd/evm}.

\bibitem{273865}
Y.~Kim, S.~Jeong, K.~Jezek, B.~Burgstaller, and B.~Scholz, ``An off-the-chain
  execution environment for scalable testing and profiling of smart
  contracts,'' in \emph{USENIX Annual Technical Conference (ATC)}, 2021, pp.
  565--579.

\end{thebibliography}

\end{document}